\documentclass{article}

\usepackage {graphicx,latexsym}
\usepackage{amsmath,amsthm}
\usepackage{amsfonts}
\usepackage{amssymb}
\usepackage[utf8]{inputenc}
\usepackage{authblk}
\usepackage{mathtools}
\usepackage{xcolor}
\usepackage{braket}
\usepackage{enumerate}
\usepackage{caption}
\usepackage{subcaption}

\newcommand{\R}{{\mathbb R}}

\newtheorem{theorem}{Theorem} 
 
\newtheorem{lemma}{Lemma} 
\newtheorem{corollary}{Corollary} 
\newtheorem{definition}{Definition} 

\DeclareMathOperator{\sgn}{sgn}

\begin{document}

\title{
A discretization of Holst's action for general relativity
}
\author{
Carlos E. Beltrán$^{1,2}$
\footnote{Corresponding Author: \ttfamily carlosbmunificacion@comunidad.unam.mx}
}
\author{
José A. Zapata$^1$
}
\affil{
$^1$ Centro de Ciencias Matemáticas, 
Universidad Nacional Autónoma de México, 
Morelia, Michoacán C.P. 58089, México. \\
$^2$ Instituto de Física, 
Universidad Nacional Autónoma de México, 
Ciudad de México C.P. 04510, México. 
}

\date{}
\maketitle

\begin{abstract} 
We present a simplicial model for gravity written in terms of a discretized Lorentz connection and a discretized tetrad field. The continuum limit of its action is Holst's action for general relativity. With the intention of using it to construct spin foam modes for quantum gravity, we write two other equivalent models written in terms of a discretized and constrained $B$ field. 
The differences between our model and existing models are most likely inessential in the sense that a quantization would lead to equivalent quantum theories in the Wilsonian continuum limit. 
Nevertheless, we mention two features leading to possible advantages: 
Curvature degrees of freedom are described at the level of each 4-simplex. 
Our model offers a picture of bulk geometry leading to actions for matter couplings that split as a sum over 4-simplices. 
\end{abstract}

\tableofcontents

\section{Motivation}
\label{Motivation}

The spin foam approach to quantum gravity \cite{Perez1, RV} builds a quantum theory starting with a version of classical general relativity. Up to now, the theory sits on a discrete structure replacing spacetime and introducing a cutoff. The hope is that a continuum limit, in the sense of Wilsonian renormalization, exists; if it does, the resulting theory would be compatible by design with kinematical canonical loop quantization in the continuum \cite{LQ}. 

It is often considered that the Holst action \cite{Holst, EPRL} for general relativity is the starting point for the construction of spin foam models for quantum gravity \cite{Perez1}. 
In order to construct it, one adds 
an additional term to the tetradic Palatini action 
\begin{equation}\label{holst}
	S_{H}(e,\omega)=\dfrac{1}{2}\int_{M}\epsilon_{IJKL}e^{I}\wedge e^{J}\wedge F^{KL}+\dfrac{1}{\gamma}\int_{M}e_{I}\wedge e_{J}\wedge F^{IJ},
\end{equation}	
where $\gamma\in\mathbb{R}\cup\left\{\infty\right\}$ is called the Barbero-Immirzi parameter, $\omega$ is a $SO(3,1)$ connection 1-form, $F$ is its curvature, $e^{I}$ is a tetrad field (a 1-form with values on the Minkowski space $\mathcal{M}$) and $\epsilon_{IJKL}$ is the standard skew-symmetric Levi-Civita tensor in $\mathcal{M}$. The first term in \eqref{holst} is the standard Palatini action. The second is an additional term that does not affect the field equations in vacuum \cite{EPRL}.%
\footnote{
Including topological terms would certainly be desirable. We briefly address the issue in the concluding section.}

A more detailed account is that the EPRL spin foam model for quantum gravity (proposed by Engle, Pereira, Rovelli and Livine) \cite{EPRL} is constructed as a quantization of a classical theory on a discretized spacetime, and the action for that classical theory resembles the Holst's action for general relativity, but it is not a direct discretization of it: 
	The main variable of the models, besides a discretized connection, is a Lie algebra valued discretized 2-form, called a $B$ field, instead of the tetrad $e$ used in the Holst action. The argument is that in the continuum a tetrad field $e$ may be traded for an appropriately constrained $B$ field. 
	What is actually done is that the $B$ field together with the mentioned constraints and the connection that live in the continuum are transformed into classical fields and constraints living on a discretized spacetime. They are later quantized in order to obtain a spin foam model.

The route back to a classical field theory does not precisely follow the path suggested by retracing the path followed in the quantization procedure. This was never a requirement; we only mention it as a curious fact. 
The classical limit is studied first, and, at a second stage, the continuum/macroscopic limit is considered: (i) The asymptotic limit of the vertex amplitude of the model (in the sense of large quantum numbers) is governed by the action of Regge calculus \cite{Regge}. (ii) The continuum/macroscopic limit of Regge's action is known to be the Einstein-Hilbert action. There are two subtleties at this point: First, contact is made with an alternative version of Regge calculus called ``area-angle'' Regge calculus \cite{AAR}. Second, at the level of computer simulations, it has been observed that the continuum limit needs a simultaneous scaling of the Planck length and the Immirzi parameter \cite{ESFM1}.

In this article, we present a discretization of general relativity in terms of a discretized connection and a discretized tetrad field $e$. We prove that the macroscopic/continuum limit of the action of our model is Holst's action for general relativity.%
\footnote{The model is also defined for degenerate tetrad fields, 
and the mentioned proof of convergence of the action holds also in the case of degenerate tetrad fields.}

We also present another model, written in terms of a discretized and constrained $B$ field and a discretized connection, which in the nondegenerate sector is equivalent to the first model. Two versions of the geometricity conditions on the discretized $B$ field are considered.%
\footnote{The $B$-field models are defined for degenerate $B$ fields. Equivalence to the tetrad field model has only been proven for nondegenerate $B$-fields and tetrad fields.} 
In the first, the main condition is entirely analogous to the so-called ``linear simplicity constraint,'' a key ingredient of the EPRL model (and of its classical prequel) \cite{EPRL}. Our proof of equivalence, however, needs supplementary constraints. 
We have reasons to expect that quantization of our model could be done emulating the EPRL strategy but in the context of our curved 4-simplices. 
The second version of the geometricity conditions is designed to be used in the quantization strategy called effective spin foam models, championed by Dittrich and collaborators \cite{ESFM1}. We think that our model could play the role now played by area-angle Regge calculus but with possible advantages. Certainly, our model provides a different point of view in several of the issues of current interest in the construction and study of models of this type. 

%

Given that the EPRL spin foam model is already constructed, that there is already a framework for effective spin foam models, and that they are related to classical simplicial models for general relativity, it is natural to ask what can the contribution of our new model be. 
We address that issue later in this article, but here we just mention how two differences between our model and the usual models based on a Regge geometry --with flat simplices, and where the direct picture of spacetime geometry is associated to the boundary of the simplices-- can be advantageous. 
Only very recently have the first spin foam simulations involving more than one simplex been performed \cite{ESFM1, SFmoreThan1simplex1, SFmoreThan1simplex2}. Since the geometry of the current spin foam models is based on flat 4-simplices, the curvature degrees of freedom have barely been tested in spin foam models for quantum gravity \cite{Livine}. In contrast, if a quantum version of our model is constructed, the study of a single 4-simplex (which we will also call spacetime atom) would necessarily include curvature degrees of freedom. This feature comes from our use of Reisenberger's discretization \cite{Rei1}. 
In contrast with most most models the geometric picture of a spacetime atom is directly associated to the interior of the 4-simplex. Although this difference can be thought of as a simple reorganization of the degrees of freedom with respect to the discretization, having a ``bulk geometry picture'' gives us the possibility of regularizing the coupling of gravity with matter fields in a simple way, which is not directly available to other models.

In section \ref{Model} we present the discrete model, show that in the continuum limit the discrete action converges to the Holst action, and deduce its field equations. In Section \ref{B} we write our model as a discrete constrained BF theory in two different ways. 
In Section \ref{NewPerspectives} we describe with some detail a few perspectives offered by our model which distinguish it from available models. 
Finally, in Section \ref{Summary} we give a brief summary of the main results 
and briefly address the issue of considering a more general action including topological terms.

\section{The model}
\label{Model}

\subsection{Discretization}
\label{discretization}

The classical discrete model presented here uses a cellular decomposition of the spacetime manifold $M$ introduced in Reisenberger in \cite{Rei1}. Start by considering an oriented simplicial complex, denoted by $\triangle$, and its corresponding dual cellular decomposition $^{*}\triangle$ as defined in \cite{Thi}. Following the notation given in \cite{Rei1}, the Greek letters $\nu, \tau, \sigma$ will denote a 4-simplex, a 3-simplex and a 2-simplex, respectively. Furthermore, 0-simplices of a given 4-simplex will be numbered from $1$ to $5$ and generically be denoted by the letters $P,Q,R,\dots$. The notation $\mu<\rho$ will mean that $\mu$ is a subcell or subsimplex of $\rho$. The notation $C_{\rho}$ will denote the barycenter of $\rho$. 

The discretization considered here is designed to contain elements of the intersection between $\triangle$ and $^{*}\triangle$. We subdivide each 4-simplex $\nu<\triangle$ into subcells called corner cells. Given a 4-simplex $\nu<\triangle$ and a vertex $P$ in $\nu$, we define the corner cell $c_{P}$ in $\nu$ associated with $P$ as the intersection between $\nu$ and the 4-dimensional cell in $^{*}\triangle$ dual to $P$. In this way, $c_{P}$ is a hypercube with vertex set composed by the points $P$, $C_{\nu}$ and the barycenters of all the subsimplices of $\nu$ that have $P$ as one of its vertices. Therefore, each 4-simplex $\nu<\triangle$ is subdivided into five corner cells, one for each vertex $P$. See Figure \ref{dis}.

Now, let us consider a 2-simplex $\sigma$ and a 4-simplex $\nu<\triangle$ containing it $\sigma<\nu$. We define the wedge $s(\sigma\nu)$ associated with $\sigma$ in $\nu$ as the intersection between the 4-simplex $\nu$ and the 2-cell in $^{*}\triangle$ dual to $\sigma$. The wedge $s(\sigma\nu)$ is a plane having as vertices the barycenters of $\nu$, $\sigma$ and the barycenters of the two 3-simplices $\tau_1,\tau_2<\nu$ that share $\sigma$ (as seen in Figure \ref{dis}). In four dimensions, there are 6 wedges in the boundary of each corner cell $c_{P}$, and 10 wedges inside each 4-simplex (without considering orientation). Inside a corner cell, a given wedge $s$ has exactly 1 dual wedge; it is the wedge that intersects $s$ only at the baricenter of the 4-simplex. Inside a 4-simplex, for each wedge there are 3 other wedges that intersect it at a single point.

Of special importance for our model are the lines $l(\nu\tau)$ resulting from the intersection between $\nu$ and the 1-cell dual to the 3-simplex $\tau<\nu$. We consider them as oriented lines starting at $C_{\nu}$ and ending at the barycenter of $\tau$. In addition, we will denote as $r(\tau\sigma)$ the oriented lines starting at the barycenter of $\tau<\nu$ and ending at the barycenter of the 2-simplex $\sigma<\tau$. 
The boundary of an oriented wedge can be written as $\partial s=l_1+r_1-r_2-l_2$. Oriented wedges are in one to one correspondence with ordered pairs $(l_1, l_2)$. 

\begin{figure}
	\centering
	\includegraphics[width=0.9\textwidth]{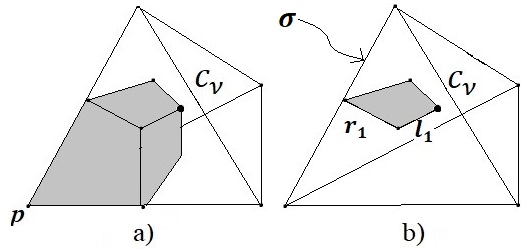}
	\caption{A corner cell and wedges in three dimensions. Panel (a) shows $c_P$, the corner cell associated with vertex $P$. Panel (b) shows $s$, the wedge dual to $\sigma$; lines $l_1, r_1$ are part of $\partial s$.}\label{dis}
\end{figure}

Now we describe the fields of our discrete model. 
Consider that there are copies of the Minkowski space $\mathcal{M}$, the Lie group $SO(3,1)$ and the Lie algebra $so(3,1)$ attached to the barycenter of every 4-simplex $\nu\in\triangle$. 
An analogous structure is also associated to the barycenter of each 3-simplex and each 2-simplex in $\nu$. 
In our conventions, capital letters from the middle of the alphabet $I, J, \ldots$ are Minkowski space indices, and lowercase letters from the middle of the alphabet $i, j, \ldots$ are Lie algebra indices. 

A discrete tetrad field $e$ is the association of a Minkowski vector $e_{l}^{I}$ to each oriented line $l(\nu\tau)$. For each corner cell in $\nu$, there are four vectors that when the field is not degenerate form a basis of the copy of Minkowski's space attached to $C_\nu$. 
Instead of a discrete connection 1-form $\omega$, we use a collection of parallel transport maps. We associate an element $h_{l}\in SO(3,1)$ with each oriented line $l(\nu\tau)$, representing the parallel transport along that curve. Analogously, we associate an element $k_{r}\in SO(3,1)$ with each oriented line $r(\tau\sigma)$. 
These are the basic fields of our model. 

Notice that the tetrad field induces a (pseudo) metric in $C\nu$ associated to each corner cell of a 4-simplex 
\begin{equation}\label{gEta}
	g_{l_1 l_2}(c) = e_{l_1}^I e_{l_2}^J \eta_{IJ}  \ \mbox{ for any } \ l_1, l_2 < c . 
\end{equation}
In addition, the metrics of neighboring corner cells $c, c'$ are clearly compatible in the sense that they if two vectors contained in the 3-dimensional subspace corresponding to $c \cap c'$ are expanded in the basis associated to $c$ or in the basis associated $c'$, the calculation of their inner product in the corresponding metrics $g(c)$ or $g(c')$ would agree. 

	Unlike what happens in other discrete models, 4-simplices here carry curvature. The discrete object measuring curvature in a 4-simplex $\nu$ is the parallel transport along the boundaries of its wedges $g_{\partial s}:=h_{l_2}^{-1}k_{r_{2}}^{-1}k_{r_1}h_{l_1}$.

\subsection{The action}
\label{action}

The action for the model can be written as a sum over 4-simplices 
\begin{equation}\label{S}
	S_{\triangle}(e,h,k)=
	\sum_{\nu<\triangle} \ \sum_{l, l', s<\nu}
	\sgn(l,l',s) \ e_{l}^{K}e_{l'}^{L}\left(T_{i}\right)^{JI}P_{IJKL} \ \theta_{s}^{i},
\end{equation}	
where a few conventions need to be clarified: 
In the sum over wedges each orientation of a wedge $s$ is considered independently. 
The sign $\sgn(l,l',s)$ is 
determined trading $s$ for its corresponding ordered pair of oriented lines $(l'',l''')(s)$; 
it is positive or negative according to the orientation of the ordered four tuple $(l, l', l'',l''')$.  
$\left(T_{i}\right)_{~~J}^{I}$ is a basis for the Lie algebra, and 
$\theta_{s}^{i} \left(T_{i}\right)_{~~J}^{I}$ is a discretization of the curvature $F_{~~J}^{I}$, where $\theta_{s}^{i}:=\text{tr}\left[T^{i}g_{\partial\bar{s}}\right]$. 
Finally, $P_{IJKL}:=\epsilon_{IJKL}-\dfrac{1}{2\gamma}\epsilon_{IJMN}\epsilon_{~~~~KL}^{MN}$. 

\begin{definition}[e-field model]
	The action $S_{\Delta}(e, h, k)$ defined in \eqref{S} determines the $e$-model. 
\end{definition}

There is a more explicit formula for $S_\nu$, the contribution to \eqref{S} from a single 4-simplex. 
The lines $l<\nu$ are labeled by numbers from 1 to 5, and written as $l_P$. Wedges are denoted by $s_{PQ}$. We simplify notation for the fields writing $e_{l_P}= e_P$ and $\theta_{s_{RS}} = \theta_{RS}$. Thus, we can write 
\begin{equation}\label{dac2}
	S_{\nu}=\sum_{T}\epsilon^{PQRST}e_{P}^{K}e_{Q}^{L}\left(T_{i}\right)^{JI}P_{IJKL}\theta_{RS}^{i}.
\end{equation}	
Notice that the sum over $T$ corresponds to a sum over corner cells in $\nu$. 

We would like to call attention to the fact that the model is defined for any discrete tetrad field, including degenerate ones. The continuum limit and the field equations presented in the following subsections also hold for general tetrad fields. 

\subsection{The continuum limit of the action}
\label{ContinuumLimit}

In this section, we show that the discrete action \eqref{S} converges in the continuum limit to Holst's action \eqref{holst}. 
This subsection follows a minor adaptation of the procedure used by Reisenberger to study the continuum limit of his model \cite{Rei1}.

Let us rewrite \eqref{holst} as 
\begin{equation}\label{holst2}
	S_{H}\left(e,\omega\right)=\int_{M}\Sigma_{i}(e)\wedge F^{i}(\omega), 
\end{equation}	
where 
$\Sigma_i:=P_{IJKL}e_{\alpha}^{K}e_{\beta}^{L}\left(T_{i}\right)^{JI}dx^{\alpha}\wedge dx^{\beta}$, and $F^{i}:=\left(T^{i}\right)_{IJ}F_{~~~\mu\nu}^{IJ}dx^{\mu}\wedge dx^{\nu}$.

\begin{definition}\label{map}
	The map $\Omega_{\triangle}:\left(\omega,e\right)\rightarrow\left(h_l,k_r,e_l\right)$ from continuum fields to simplicial fields defined on the discretization of $M$ is given by the following expressions:
	\begin{align}\label{map2}
		h_{l}&:=Pexp\left(\int_{l}\omega\right)\\
		k_{r}&:=Pexp\left(\int_{r}\omega\right)\\
		e_{l}^{I}&:=\int_{l}U_{~~J}^{I}(x,C_{\nu})e^{J}(x).
	\end{align}	
\end{definition}	 
\noindent where $Pexp$ is the path-ordered exponential and $U_{~~J}^{I}(x,C_{\nu}):=Pexp\left(\int_{C_{\nu}}^{x}\omega\right)$ is the parallel transport of Minkowski vectors along a straight line, from $x$ to $C_{\nu}$. 

In what follows, we will work with compact spacetimes or compact regions of spacetime. Then simplicial decompositions are finite. To prove convergence, we use a sequence of triangulations which become finer everywhere at the same rate. This idea is formalized by the following definition. 
\begin{definition}\label{unif}
	Let $\left\{\triangle_n\right\}_{n=0}^{\infty}$ be a sequence of simplicial decompositions of a compact manifold $M$, and $g_0$ a fixed positive-definite metric constant on each $\nu<\triangle_0$. We say that $\left\{\triangle_n\right\}_{n=0}^{\infty}$ is a uniformly refining sequence if 
	\begin{enumerate}[a)]
		\item $\triangle_0$ is a finite simplicial decomposition of $M$.
		\item $\triangle_{n+1}$ is a finite refinement of $\triangle_{n}$, $\forall n$.
		\item $\lim_{n\rightarrow\infty}r_n=0$, with $r_n$ the maximum of the $g_0-$radii $r_\nu$ of the 4-simplices $\nu<\triangle_n$.
		\item $r_{\nu}^{4}$ divided by the $g_0$ 4-volume of $\nu$ is uniformly bounded as $n\rightarrow\infty$, for all 4-simplices $\nu<\triangle_{n}$.
	\end{enumerate}	 
\end{definition}

We also need two preliminary lemmas. The first is an elementary calculus result.
\begin{lemma}\label{lemma1}
	Let $x,y,x_0,y_0\in\mathbb{R}$ such that 
	\begin{align*}
		\left|x-x_0\right|&<\text{min}\left(1,\dfrac{\epsilon}{2\left(\left|y_0\right|+1\right)}\right)\\
		\left|y-y_0\right|&<\dfrac{\epsilon}{2\left(\left|x_0\right|+1\right)},
	\end{align*}
	for some $\epsilon>0$. Then $\left|xy-x_0y_0\right|<\epsilon$.	
\end{lemma}

Given a k-subcell c of a 4-simplex $\nu$, let us define the characteristic tensor of $c$ as
\begin{equation}\label{chac}
	t_{c}^{\alpha_1,\dots,\alpha_k}:=\int_{c}dx^{\alpha_1}\wedge\dots\wedge dx^{\alpha_k} , 
\end{equation}	
where $x^{\alpha}$ are linear coordinates on $\nu$. Expression \eqref{chac} is useful for expressing the result of the following lemma.
\begin{lemma}\label{lemma2}
	Let $h$ be a continuous k-form on a compact manifold $M$, $\left\{\triangle_n\right\}_{n=0}^{\infty}$ be a uniformly refining sequence of simplicial decompositions of $M$, and $g_0$ be a positive definite metric that is constant on each 4-simplex $\mu<\triangle_{0}$. Then for all $\epsilon>0$ there exists $N>0$ such that for any $n>N$ and any k-subcell $c<\nu<\triangle_{n}$ we have 
	\begin{equation}\label{aux1}
		\left|\int_{c}h-h\left(C_{\nu}\right)_{\alpha_1,\dots,\alpha_k}t_{c}^{\alpha_1,\dots,\alpha_k}\right|<\epsilon r_{\nu}^{k},
	\end{equation}
	where $r_{\nu}$ is the $g_0$ radius of $\nu$ and $h(C_{\nu})_{\alpha_1,\dots\alpha_k}$ are the components, in linear coordinates, of $h$ evaluated in $C_{\nu}$.	
\end{lemma}
The proof of Lemma \ref{lemma2} can be found in \cite{Rei1}. 

Now we are prepared to prove the desired convergence result. 
\begin{theorem}\label{theo1}
	Let $M$ be a four dimensional compact oriented manifold and $\left\{\triangle_{n}\right\}_{n=0}^{\infty}$ a sequence of uniformly refining simplicial decompositions of $M$. Additionally, let $(\omega, e)$ be fields of the Holst action (with $e$ continuous and $\omega$ continuously differentiable) and $S_{n}(e,\omega)$ the evaluation of the discrete action \eqref{S} on the simplicial fields $(h_l,k_r,e_l)_{n}=\Omega_{\triangle_n}(e,\omega)$ defined on $\triangle_n$ using Definition \ref{map}. Then 
	\begin{equation}\label{conlim}
		\lim_{n\rightarrow\infty}S_{\triangle_n}(\omega,e)=S_{H}(\omega,e) .
	\end{equation}	
\end{theorem}
\begin{proof}
	Let 
	$g_0$ be a positive-definite metric that is constant on each $\nu<\triangle_0$ with respect to its linear coordinates and $\epsilon>0$. From Lemmas \ref{lemma1} and \ref{lemma2} it follows that there exists $N>0$ such that for any $n>N$, $\nu<\triangle_{n}$ and $s, l_1, l_2 < \nu$ 
	\begin{align}
		\left|\theta_{s}^{i}-F_{\alpha\beta}^{i}(C_{\nu})t_{s}^{\alpha \beta}\right|&<\epsilon r_{\nu}^{2}\\
		\left|e_{l_1}^{K}e_{l_2}^{L}\left(T_{i}\right)^{JI}P_{IJKL}-\Sigma_{i,\alpha\beta}(C_{\nu})t_{l_1}^{\alpha}t_{l_2}^{\beta}\right|&<\epsilon r_{\nu}^{2}.
	\end{align}	
	In this way, for any $n>N$ 
	\begin{equation}\label{ev}
		S_{\triangle_n}=4\sum_{\nu<\triangle_n}\left[\Sigma_{i,\alpha\beta}(C_{\nu})F_{\gamma\delta}^{i}(C_{\nu})
		\sum_{l_1,l_1,s<\nu}t_{l_1}^{\alpha}t_{l_2}^{\beta}t_{s}^{\gamma \delta}\sgn(l_1, l_2, s)+\Delta I_{n_{\nu}}\right],
	\end{equation}	
	with a bounded error term $\Delta I_{n_{\nu}}$: 
	\begin{equation}\label{bound}
		\left|\Delta I_{n_{\nu}}\right|<480\epsilon r_{\nu}^{4}\left[\epsilon+\left\|\Sigma\right\|_{C_\nu}+\left\|F\right\|_{C_\nu}\right],
	\end{equation}	
	where the norm $\left\|T\right\|$ of a tensor $T_{\alpha_1,\dots,\alpha_n}^{i_1,\dots,i_m}$ with both Lie algebra and spacetime indices is 
	\begin{equation}\label{norm}
		\left\|T\right\|=\left[T_{\alpha_1,\dots,\alpha_n}^{i_1,\dots,i_m}T_{\beta_1,\dots,\beta_n}^{j_1,\dots,j_m}g_{0}^{\alpha_1\beta_1}\cdots g_{0}^{\alpha_n\beta_n}\lambda_{i_1j_1}\cdots\lambda_{i_mj_m}\right]^{1/2},
	\end{equation}	
	with $\lambda_{ij}$ the $so(3,1)$ Killing-Cartan metric.

Since the fields $\Sigma$ and $F$ are continuous on $M$, which is assumed to be compact, they are bounded. We can then write
	\begin{equation}\label{bound2}
		\left|\Delta I_{n_{\nu}}\right|<\epsilon r_{\nu}^{4}k,
	\end{equation}
	with $k$ a finite positive constant. Condition d) in Definition \ref{unif} says that the ratio $r_{\nu}^{4}/V_{\nu}$ (where $V_{\nu}$ is the $g_0$ volume of $\nu$) is uniformly bounded. In other words, $r_{\nu}^{4}<RV_{\nu}$ for some $R>0$. We then have that $\left|\Delta I_{n_{\nu}}\right|<\epsilon RkV_{\nu}$, and the sum of errors is bounded by
	\begin{equation}\label{bound3}
		\left|\sum_{\nu<\triangle_n}\Delta I_{n_{\nu}}\right|<\epsilon RV_{M}k,
	\end{equation}
	where $V_{M}$ is the $g_0$ volume of $M$. 
	
	On the other hand, it is possible to show that
	\begin{equation}\label{ident}
		\sum_{PQRST}t_{l_P}^{\alpha}t_{l_Q}^{\beta}t_{s_{RS}}^{\gamma\delta}\epsilon^{PQRST}=\dfrac{1}{4}t_{\nu}^{\alpha\beta\gamma\delta}.
	\end{equation}	
	By using all the results given above, we obtain 
	\begin{equation}\label{result}
		\lim_{n\rightarrow\infty}S_{\triangle_n}=\lim_{n\rightarrow\infty}\left(\left[\sum_{\nu<\triangle_n}\Sigma_{i,\alpha\beta}(C_{\nu})F_{\gamma\delta}^{i}(C_{\nu})t_{\nu}^{\alpha\beta\gamma\delta}\right]+\Delta I_{n}\right)=\int_{M}\Sigma_{i}\wedge F^{i}.
	\end{equation}	
	
\end{proof}

\subsection{Field equations}
\label{field}

Before starting the calculations, we will further simplify the expression for the action \eqref{S} and point out a feature of our discrete fields that is not present in closely related discrete models for gravity like the classical EPRL model \cite{EPRL}. 

Let us introduce $\hat{J}_s$, a function of the tetrad field $e_{l}$, and $J^{s}$, a function of $\hat{J}_s$. 
\begin{equation}\label{jtilde}
	\hat{J}_{s_{PQ}~i}:= e_{l_P}^{K}e_{l_Q}^{L}\left(T_{i}\right)^{JI}P_{IJKL}.
\end{equation} 
\begin{equation}\label{Bf}
	J^s_{i}:= \sum_{s' <\nu} \sgn(s', s) \ \hat{J}_{s'~i} .
\end{equation} 
We can also write $J^{PQ}_{i} := \sum_T \epsilon^{PQRST} \hat{J}_{{RS}~i}$. 

By using \eqref{Bf}, the action \eqref{S} can be written as
\begin{equation}\label{dac}
	S_{\triangle}(e, h, k)=\sum_{\nu<\triangle}\left[\sum_{s<\nu} J^s_{i}\theta_{s}^{i}\right].
\end{equation}	
We then discover an identity of our discrete model and something that looks like a symmetry of the model but is not so. 
They follow from the fact that the function taking the $\hat{J}_s$ variables to the $J^{s}$ is a linear transformation from $so(3,1)^{10}$ to itself which is not 1 to 1. The image of the transformation is the solution space of the relations 
\begin{equation}\label{closure}
	\sum_{s >l_P} J^s_{i} = 0 , 
\end{equation}
one such relation for each $l_P>\nu$. 
That is, if we consider the previous relations as a function of the $\hat{J}_s$ variables they become identities; they could be called  ``closure identities.'' The calculation is simple: $\sum_{s >l_P} J^s_{i} = 
	\sum_{TQ} \epsilon^{PQRST} \hat{J}_{{RS}~i} = 0$. 
On the other hand, the kernel of the transformation is given by elements of the form $K_{PQ~i}=\left(\delta_{P}^{I}-\delta_{Q}^{I}\right)k_{I~i}$, with $k_{I~i}\in so(3,1)$. 

Our basic fields are not the $\hat{J}$ variables but the tetrad field $e$, but let us consider \eqref{dac} as a function of the discrete connection and $\hat{J}$. In that case, the transformation $\hat{J}_{s~i}\rightarrow \hat{J}'_{s~i} = \hat{J}_{s~i}+K_{s~i}$ would be a gauge symmetry. 
For our model, however, the mentioned transformation is not a symmetry because it generically takes $\hat{J}$ variables that can be calculated from a tetrad $e$ following \eqref{jtilde} to variables that cannot be expressed in the same way.%
\footnote{If we were considering a discretization of BF theory with $\hat{J}$ as a basic field, the mentioned transformation would be a gauge symmetry of the model.}  

Now we begin calculating the field equations. They are obtained by taking the variation of \eqref{S} or \eqref{dac} with respect to the fields $h_{l}$, $k_{r}$ and $e_{l}$, while maintaining the variables $k_{r}$ fixed at the boundary of the region.

The variation of \eqref{dac} with respect to $h_{l}$ can be calculated by parametrizing the variations of this variable in the form $h_{l}+\delta h_{l}=h_{l}\text{exp}(\lambda\alpha_{l}^{i}T_{i})$, where $\lambda\in\mathbb{R}$ and $\alpha_{l}^{i}T_{i}\in so(3,1)$. For any $s$ with $l < s$, this induces a variation on $\theta_{s}^{i}$ given by
\begin{equation}\label{varth}
	\delta\theta_{s}^{i}:=\dfrac{d\theta_{s}^{i}}{d\lambda}\bigg\rvert_{\lambda=0}=\alpha_{l}^{j}\text{tr}\left[T_{j}T^{i}g_{\partial s}\right]. 
\end{equation}	
Therefore, the variation of \eqref{dac} with respect to $h_{l}$ is given by
\begin{equation}\label{var1}
	\delta S_{\triangle}:=\dfrac{d S_{\triangle}}{d\lambda}\bigg\rvert_{\lambda=0}=\alpha_{l}^{j}\sum_{s>l}w^s_{j},
\end{equation}
with	
\begin{eqnarray}\label{ws}
	w^s_{j}=&
	&J^s_{i} \ \text{tr}\left[T_{j}T^{i}g_{\partial s}\right] \nonumber \\ 
	=&
	&\left( \sum_{l_1, l_2 <\nu}\sgn(l_1, l_2, s)e_{l_1}^{K}e_{l_2}^{L}\right) \left(T_{i}\right)^{JI}P_{IJKL}\text{tr}\left[T_{j}T^{i}g_{\partial s}\right]
	. 
\end{eqnarray}

Thus, the extremality condition with respect to variations in $h_{l}$ is 
\begin{equation}\label{feq1}
	\sum_{s>l}w^s_{j}=0 . 
\end{equation}
%

Notice that when $g_{\partial s}= \mbox{id}$, due to \eqref{closure}, this ``Gauss law'' reduces to an identity. Once the tetrad field is given, the ``Gauss law'' demands that the discretized connection be frame-compatible.

The variation of the action with respect to $k_{r}$ can be calculated by writing $k_{r}+\delta k_{r}=\text{exp}(\lambda\beta_{r}^{j}T_{j})k_{r}$. Since $g_{\partial s}=h_{l_2}^{-1}k_{r_2}^{-1}k_{r_1}h_{l_1}$, then the variation of $\theta_{s}^{i}$ is given by
\begin{displaymath}
	\delta\theta_{s}^{i}=\beta_{r_1}^{j}\text{tr}\left[T^{i}h_{l_2}^{-1}k_{r_2}^{-1}T_{j}k_{r_1}h_{l_1}\right],
\end{displaymath}	 
and the variation of \eqref{dac} is
\begin{equation}\label{var2}
	\delta S_{\triangle}=\alpha_{r_1}^{j}u_{\sigma\nu_1~j}+\alpha_{r_1}^{j}u_{\sigma\nu_2~j} , 
\end{equation}	
where $\nu_1$ and $\nu_2$ are the two 4-simplices that share the line $r_1$, and $u_{\sigma\nu~j}:=k_{r_1}h_{l_1}w^{s(\sigma \nu)}_jh_{l_1}^{-1}k_{r_1}^{-1}$.

Extremization of $S_{\triangle}$ with respect to $k_{r_1}$ implies 
\begin{equation}\label{feq2}
	u_{\sigma\nu_1}=u_{\sigma\nu_2} . 
\end{equation}	
Notice that $u_{\sigma\nu}$ is obtained by the parallel transport of $w^{s(\sigma\nu)}$ from $C_{\nu}$ to $C_{\sigma}$ following the lines $l_{1}, r_{1}$ in a positive direction. Field equation \eqref{feq2} implies that the variables $u_{\sigma\nu,j}$ are independent of the 4-simplices $\nu$ and are then uniquely associated with the 2-simplices $\sigma$ in $\triangle$. 

Finally, extremization with respect to the field $e_{l_P}^K$ yields
\begin{equation}\label{feq3}
	\sum_{T}\epsilon^{PQRST}e_{Q}^{L}\left(T_{i}\right)^{JI}P_{IJKL}\theta_{RS}^{i} = 0. 
\end{equation}

Equations \eqref{feq1}, \eqref{feq2} and \eqref{feq3} are the field equations that define the dynamics of our model. 

If the spacetime region considered has a boundary, the contributions to the action from 4-simplices meeting the boundary may need to be modified for the purpose of arriving to the correct bulk field equations without any equations conflicting with the boundary conditions. 
Clearly, if the gauge field $k_{r}$ is kept fixed at the boundary, no modification is needed. On the other hand, if its momentum, the variable $u_{\sigma}$, is kept fixed at the boundary, we need to avoid conflict between field equation \eqref{feq2} and the boundary condition. Consider that $r_1<s$ is part of the boundary; then in the action we should replace 
$\theta_{s}^{i}:=\text{tr}\left[T^{i} h_{l_2}^{-1}k_{r_{2}}^{-1}k_{r_1}h_{l_1} \right]$ with 
$\tilde{\theta}_{s}^{i}:=\text{tr}\left[T^{i} h_{l_2}^{-1}k_{r_{2}}^{-1} \text{id} h_{l_1} \right]$. In \cite{Rei1}, the analogous modification to the action at boundaries was proposed.

\section{Formulation in terms of a constrained $B$ field}
\label{B}

The discrete $B$ field assigns an element of the copy of $\mathfrak{so}(3,1)$ associated with $C_\nu$ to every wedge $s<\nu$. We write it as $B_s = B_{si} T^i \in \mathfrak{so}(3,1)$. 
Such a field may be considered the result of integrating the $B$ field in the continuum over $s$, in a way that is compatible with gauge transformations following a prescription similar to \eqref{map2}. 

The degrees of freedom inside each spacetime atom (4-simplex) can be neatly divided into bulk degrees of freedom and boundary degrees of freedom. The $B$ field and the $h$ field live in the bulk, and the only basic field living at the boundary is the boundary connection $k$. 

Our starting point will be the action 
\begin{equation}\label{B-action}
	S_{\Delta}(B, h, k)=\sum_{\nu<\Delta}\sum_{s,s'<\nu}\text{sgn}(s,s')\left( B_{s i}-\dfrac{1}{\gamma}{\ast}B_{s i} \right) \theta^{i}_{s'}  
\end{equation}	
together with a set of constraints on the $B$ field that makes it a discrete model for gravity. 

The idea is that if $B_{PQ} = \ast \ e_P \wedge e_Q$, then the action \eqref{B-action} reduces to \eqref{S}%
\footnote{Notice that $\hat{J}_{s,i}=B_{s,i}-\dfrac{1}{\gamma}{\ast}B_{s,i}$, but in this section $B$ is a basic field and not just a convenient variable.}. 
The contribution to the action from a 4-simplex can also be written as a sum over corner cells 
\[
S_\nu = \sum_T \epsilon^{PQRST}
\left( B_{PQ i}-\dfrac{1}{\gamma}{\ast}B_{PQ i} \right) \theta^{i}_{RS} . 
\]

\subsection{Geometricity à la EPRL}\label{GeometricityEPRL}

We will introduce a series of geometricity/simplicity conditions, formally stated below as \ref{G1}, \ref{G2} and \ref{G3}, 
which by construction will ensure that $B_{PQ} = \ast \ e_P \wedge e_Q$. 

\begin{definition}[Condition G1($\nu$)]\label{G1}
The $B$ field in spacetime atom $\nu$ is such that 
for any internal edge $l < \nu$ there is a vector $\vec{n}_l \in \R^{3,1} - \{ \vec{0} \}$ such that for any wedge $s$ with $l<s<\nu$ we have 
\[
n_{l I} B_s^{IJ}=0 .
\]	
\end{definition}

\noindent 
This constraint is analogous to the so-called ``linear simplicity'' constraint \cite{FK}, which is at the core of the construction of the EPRL model \cite{EPRL}. 
An important difference is that at each atom our constraint involves bulk degrees of freedom only.

Once we have a $B$ field satisfying Condition \ref{G1}, let us fix a choice of set of vectors $\{ n_l \}_{l<\nu}$ participating in \ref{G1}; 
it will be clear later that our considerations are independent of this choice. Condition \ref{G1} holds if and only if the $B$ field inside $\nu$ can be written as $B_{PQ} = \lambda_{PQ} \ast n_P \wedge n_Q$ for a set of scalars $\{\lambda_{PQ}=\lambda_{QP}\}$. 

If we had a ``tetrad'' written as $e_P = \lambda_P n_P$ for some coefficients $\lambda_P$, then the relation between $B$ and $e$ induces the system of equations 
\begin{equation}\label{lambdaEq1}
\lambda_{PQ} = \lambda_P \lambda_Q .	
\end{equation}
In this system of equations, we consider $\lambda_{PQ}$ as given data and $\lambda_P$ as unknown parameters that will determine the ``tetrad.'' For being able to determine $\lambda_P$, we add a second set of conditions on $B$. 

Recall that two wedges $s,s' < c$ contained in the same corner cell are transversal if and only if $\text{sgn}(s,s') \neq 0$. 

\begin{definition}[Condition G2($c$)]\label{G2}
	For any two pairs of transversal wedges $(s,s')$, $(s'',s''')$ {\bf contained in the same corner cell} $c$ the following equation holds 
\[
	\text{sgn}(s,s')\epsilon_{IJKL}B_s^{IJ}\wedge B_{s'}^{KL}=\text{sgn}(s'',s''')\epsilon_{IJKL}B_{s''}^{IJ}\wedge B_{s'''}^{KL}. 
\]
\end{definition}
This condition states that the 4-volume of the corner cell $c$ may be computed as $V(c):= 4^4 \text{sgn}(s,s')\epsilon_{IJKL}B_s^{IJ}\wedge B_{s'}^{KL}$ for any pair of transversal wedges $(s,s')$ contained in  $c$.%
\footnote{
The 4-volume of the corner cell $c<\nu$ may be calculated by dividing the corner cell into $4!$ ``flags'' (4-simplices of the baricentric subdivision of $\nu$). The volume of any flag in corner cell $c_5$ is $4^3/6$ times the volume of the parallelepiped determined by $l_1, l_2, l_3, l_4$.} 
The field is considered to be nondegenerate at $c$ if $V(c) \neq 0$.  
Conditions of this type have also previously appeared in discrete gravity models \cite{FD, Rei1}. In addition, a similar condition is considered as part of the classical EPRL model \cite{EPRL}. 
In the quantum EPRL model, however, the analogous condition is automatically satisfied. For a discussion, see \cite{Perez1}. 

For the sake of concreteness, let us examine the conditions for corner cell $c_{5}$ and study the case of a nondegenerate field satisfying Conditions G1 and G2. Condition G1 implies that the normal vectors $n_1, n_2, n_3, n_4$ are defined. 
Non degeneracy is equivalent to $\epsilon_{IJKL} n_1^I n_2^J n_3^K n_4^L \neq 0$, 
and Condition G2 may be stated as 
\begin{equation}\label{lambdaEq2}
	\epsilon^{PQRS5} \neq 0 \ \ \ \mbox{ implies } \ \ \ 
	\lambda_{PQ} \lambda_{RS} = \lambda_{12} \lambda_{34} \doteq \Lambda_5 . 
\end{equation}

A set of independent coefficients generating the rest (in the nondegenerate case) is $\lambda_{(12)}$, $\lambda_{(13)}$, $\lambda_{(14)}$, $\Lambda_5$. 
The system of equations \eqref{lambdaEq1}, \eqref{lambdaEq2} is easily solvable. We define 
$\lambda_1^2 \doteq \frac{\lambda_{13} \lambda_{14}}{\lambda_{34}} = \frac{\lambda_{12} \lambda_{13} \lambda_{14}}{\Lambda_5}$. In general, we can write 
$\lambda_P^2 \doteq \frac{\lambda_{P P+1} \lambda_{P P+2} \lambda_{P P+3}}{\Lambda_5}$ using appropriate cyclic conventions. 
It is clear that if $P, Q \in \{ 1,2,3,4 \}$ are different then 
$\lambda_P^2 \lambda_Q^2 = (\lambda_{PQ})^2$. Thus, if all the coefficients $\lambda_P^2$ are positive, we have found real solutions to the system of equations that let us determine $e$ in terms of $B$ for the variables associated with the corner cell $c_{5}$. 

An example showing that the positivity requirements are not trivial is constructed staring with a $B$ that is of the form $B_{PQ} = \ast e_P \wedge e_Q$ and defining $B' = -B$. This new field satisfies Condition \ref{G1} (with the same vectors $n_P$) and also satisfies Condition \ref{G2}. 
On the other hand, since $\lambda'_{PQ}= - \lambda_{PQ}$ the abovementioned positivity requirements are not met. 

Associated with cell $c_{5}$, we have positivity requirements for $\lambda_1^2$, $\lambda_2^2$, $\lambda_3^2$ and $\lambda_4^2$, but since $P\neq Q$ implies $\lambda_P^2 \lambda_Q^2 = (\lambda_{PQ})^2$, either all the coefficients $\lambda_P^2$ are positive or all of them are negative. 

Below we will state a condition, written directly in terms of the $B$ field, which is equivalent to the above positivity conditions. 
Since $\lambda_1^2$, when the $B$ field is geometric, it has the interpretation of being the proportionality factor between $n_1 \cdot n_1$ and $e_1 \cdot e_1$, it is not hard to guess that the appropriate expression is reminiscent of the $(1, 1)$ component of Urbantke's metric \cite{Urbantke}. 
A straightforward calculation leads to 
\[
	\frac{1}{12 V(c_{5})}f_{ijk} \epsilon^{ABCD5} B^i_{1A} B^j_{BC} B^k_{D1} = \lambda_1^2 \ (\eta_{IJ} n_1^I  n_1^J).
\]
Since the left-hand side of this expression is independent of the choice of vectors $n_P$, two comments follow: 
First, the positivity conditions are independent of the choice of vectors $n_P$. Let us explain. 
Assume that conditions \ref{G1} and \ref{G2} hold. 
A different choice of the vectors $\tilde{n}_P = k_P n_P$ for some scalars $k_P \neq 0$ (which is all the freedom of choice, once we are given a nondegenerate $B$ field satisfying Condition \ref{G1}) leads to a right hand side of the form $\tilde{\lambda}_1^2 \ (\eta_{IJ} \tilde{n}_1^I  \tilde{n}_1^J)= \tilde{\lambda}_1^2 \ k_1^2 \ (\eta_{IJ} n_1^I  n_1^J)$. Thus, $\tilde{\lambda}_1^2 > 0$ if and only if $\lambda_1^2 > 0$. 
Second, since the coefficients $\lambda_P^2$ are either all positive or all negative, the two available options lead to bilinear forms of different signatures. This observation leads to our last geometricity condition.

\begin{definition}[Condition G3($c$)]\label{G3}
	At each corner cell we have a metric density, constructed following Urbantke, 
\begin{equation}\label{UrbMetric}
	\tilde{g}_{QR}(c_{P}) =  
	\frac{1}{12}f_{ijk} \epsilon^{ABCDP} B^i_{QA} B^j_{BC} B^k_{DR} . 
\end{equation}
The condition is that the corresponding bilinear form 
$g_{QR}(c_{P}) = \frac{1}{V(c_{P})} \tilde{g}_{QR}(c_{P})$ has the same signature as Minkowski's metric $\eta$. 
\end{definition}

We have constructed a system of equations with an explicit solution corresponding to a four tuple of vectors $e_l$ for each of the five corner cells determined by the the restriction of the discrete $B$ field to that corner cell. 
If the solutions from different corner cells were compatible we would have a way to find a discrete tetrad field $e$ in terms of a discrete $B$ field. 
Consider the solutions found for $\lambda_1^2$ associated with corner cells $c_5$ and $c_4$. 
For corner cell $c_5$ we wrote, 
$\lambda_1^2(c_5) = \frac{\lambda_{12} \lambda_{13} \lambda_{14}}{\Lambda_5}$. 
Following the same procedure for $c_4$ we obtain 
$\lambda_1^2(c_4) = \frac{\lambda_{12} \lambda_{13} \lambda_{15}}{\Lambda_4}$ (with $\Lambda_4 = \lambda_{12} \lambda_{35}$). 
A simple calculation shows that 
$\frac{\lambda_1^2(c_5)}{\lambda_1^2(c_4)}= \frac{\lambda_{35}\lambda_{14}}{\lambda_{34}\lambda_{15}} = \frac{\Lambda_2}{\Lambda_2}=1$, which proves compatibility. 

We have arrived constructively to the main result of this subsection. 
\begin{theorem}\label{Thm1}
	A nondegenerate $B$ field may be written in terms of a tetrad $e$, which is unique up to sign, as 
	$B_{PQ} = \ast e_P \wedge e_Q$ 
	if and only if conditions \ref{G1}, \ref{G2} and \ref{G3} hold. 
\end{theorem}

It is not difficult to see that under the conditions of the above theorem the pseudometrics \eqref{gEta} and \eqref{UrbMetric} agree. Then if $c$ and $c'$ are neighboring corner cells in a 4-simplex, the corresponding metrics are compatible as stated in Subsection \ref{discretization}. 

Now we will define another discrete model. 
\begin{definition}[B-field model 1]
	The action $S_{\Delta}(B, h, k)$ defined in \eqref{B-action}, together with Conditions \ref{G1} and \ref{G2} restricting the $B$ field, determines $B$-model 1. 
\end{definition}

A discrete tetrad field $e$ is nondegenerate at a corner cell $c$ if and only if $\{e_P\}_{l_P<c}$ is a basis for the copy of Minkowski's space sitting at $c$. It is clear that if a tetrad field and a $B$ field are related by $B_{PQ} = \ast e_P \wedge e_Q$ then $B$ is nondegenerate if and only if $e$ is nondegenerate. 

\begin{corollary}\label{Cor1}
$B$-model 1, in its nondegenerate sector and complemented with Condition \ref{G3}, is equivalent to the $e$-model in its nondegenerate sector. 
\end{corollary}

%
%
%
%
%

\subsection{Alternative form of the geometricity conditions}
\label{AlternativeGeoCond}

First, a set of quadratic geometricity conditions GG1 (strengthening condition G2) is imposed at each corner cell. 

\begin{definition}[Condition GG1($c$)]\label{GG1}
For all $s, s' < c$ 
\[
	\epsilon_{IJKL}B_s^{IJ}\wedge B_{s'}^{KL} - \frac{1}{12} \text{sgn}(s,s')
	\sum_{s'', s''' <c}
	\text{sgn}(s'',s''')\epsilon_{IJKL}B_{s''}^{IJ}\wedge B_{s'''}^{KL} = 0 . 
\] 
\end{definition}

A wellknown result of De Pietri and Freidel \cite{FD} (see also \cite{Thi}), translated to our Lorentzian case, says that the space of solutions of GG1 is divided into sectors: In the sector of interest, we have (i) $B_{PQ} = \ast e_P \wedge e_Q$;  other sectors of solutions are composed of $B$ fields of the form (ii) $B_{PQ} = e_P \wedge e_Q$, (iii) $B_{PQ} = - e_P \wedge e_Q$ and (iv) $B_{PQ} = - \ast e_P \wedge e_Q$. We will impose constraints GG2 and GG3, which discard the unwanted sectors. 

Clearly, condition G2 from the previous subsection is included in the set of constraints GG1. Thus, the 4-volume of a corner cell may be calculated as $V(c):= 4^4 \text{sgn}(s,s')\epsilon_{IJKL}B_s^{IJ}\wedge B_{s'}^{KL}$ for any pair of transversal wedges $(s,s')$ contained in $c$. 
\begin{definition}[Condition GG2($c$)]\label{GG2}
The $B$ field in corner cell $c$ is such that $V(c) > 0$. 
\end{definition}

\begin{definition}[Condition GG3($c$)]\label{GG3}
	At each corner cell, 
	the bilinear form defined in terms of the $B$ field by 
\[
	g_{QR}(c_{P}) =  
	\frac{1}{12 V(c_{P})}f_{ijk} \epsilon^{ABCDP} B^i_{QA} B^j_{BC} B^k_{DR}
\]
has the same signature as Minkowski's metric $\eta$. 
\end{definition}

Now we can state the main result of this subsection. 
\begin{theorem}\label{Thm2}
	A nondegenerate $B$ field may be written in terms of a tetrad $e$, which is unique up to sign, as 
	$B_{PQ} = \ast e_P \wedge e_Q$ 
	if and only if conditions \ref{GG1}, \ref{GG2} and \ref{GG3} hold. 
\end{theorem}

\begin{proof}
	First, a simple calculation shows that a field of the form $B_{PQ} = \ast e_P \wedge e_Q$ satisfies conditions GG1, GG2 and GG3. Let us see that the set of conditions are sufficient. 
	
	The mentioned result of \cite{FD} states that the set of solutions of GG1 is composed by sectors (i) to (iv) listed above. 
	Since (i) satisfies condition GG2 and $V(c)[\ast B] = - V(c)[B]$, we see that condition GG2 discards sectors (iii) and (ii). On the other hand, since condition GG3 is cubic and it is satisfied by sector (i), it rules sector (iv) out. 
\end{proof}

Now we introduce yet another discrete model. 
\begin{definition}[B-field model 2]
	The action $S_{\Delta}(B, h, k)$ defined in \eqref{B-action}, together with Condition \ref{GG1}  restricting the $B$ field, determines $B$-model 2. 
\end{definition}

\begin{corollary}\label{Cor2}
$B$-model 2, complemented with Conditions \ref{GG2} and \ref{GG3}, is equivalent to the $e$-model in its nondegenerate sector. 
\end{corollary}

\section{New perspectives}\label{NewPerspectives}

Differences among models obtained from a field theory by a cutoff are to be expected, but in a quantum continuum limit inessential differences are expected to be washed out by universality. From this point of view, the differences that we will be talking about in this section are most likely inessential. 
Among equivalent models, the possible advantages of one of them over another have two possible origins. Either one model leads to a certain insight that is not as clear in the other model, or one model leads to more efficient simulations (in general or when a particular observable is studied). 

Below we comment on two possible advantages of our model. The first one follows from the use of curved spacetime atoms. The second one is that the ``bulk geometry picture'' of our spacetime atoms leads to a natural strategy to couple matter fields which is not present for other models. 

We also comment on difference between our model and most other models, which after a more detailed consideration can be regarded as apparent (or superfluous) differences only. 

\subsection{Apparent differences}


It is evident that the contribution to the action from a 4-simplex in our model is more complicated than the terms of the action in the classical EPRL model. Considering that our 4-simplices include curvature (in a complete set of directions for the discrete geometry), what we should compare, however, is the complexity of the action of one of our 4-simplices with the action associated with the star of a vertex in the triangulation of a Regge discretization. From this perspective, the action of our model does not seem to be overly complicated; the difference seems to be mostly how degrees of freedom are organized with respect to the simplicial decomposition. Below we see two other apparent differences with the same origin.

In Reisenberger's original proposal \cite{Rei2}, a spin foam model may be constructed after integrating over the bulk degrees of freedom of spacetime atoms. It is then natural to ask how does our model look like when described in terms of boundary data. 
In our classical model, the set of boundary data is composed of the boundary connection and its momenta 
$u_\sigma$``$= \frac{\delta S}{\delta k_r} = - \frac{\delta S}{\delta k_{r'}}$'' 
(where $\partial s(\sigma \nu) = l + r - r' - l'$). The formula in terms of the tetrad field is 
\[
u_{\sigma\nu~j}:=k_{r_1}h_{l_1} \left(
\sum_{l_1, l_2 <\nu}\sgn(l_1, l_2, s)e_{l_1}^{K}e_{l_2}^{L}\left(T_{i}\right)^{JI}P_{IJKL}\text{tr}\left[T_{j}T^{i}g_{\partial s}\right]
\right) h_{l_1}^{-1}k_{r_1}^{-1} . 
\]
Notice that the momentum variable associated with a wedge $s$ is an $so(3,1)$ element which is not simple in the generic case; tat is,  it cannot be written as a wedge product of two Minkowski vectors. This looks like a clear difference with the classical prequel of the EPRL model \cite{EPRL}. At the level of corner cells boundaries, however, simplicity is present: 
Consider a corner cell $c<\nu$ and a wedge $s(\sigma \nu)$ intersecting it. 
Let us define $u_{\sigma \nu}|_c$ as the restriction of the sum given above where 
$\sum_{l_1, l_2 <\nu}$ is replaced by $\sum_{l_1, l_2 < c}$. In the resulting sum, there are only two nonvanishing terms corresponding to the wedge product of two Minkowski vectors. 
Thus, the apparent difference arises because in our spacetime atoms we have momentum variables which are a sum of simple momenta associated to the corner cells.

The signature geometry of the EPRL model is twisted boundary geometry \cite{FS}. It is intriguing that, even in our B-models, this key feature seems to be absent. A more detailed study reveals that this is only an apparent difference even at the classical level. 
In our B-models, the $B$ field is shared by neighboring corner cells within the same 4-simplex, and this gluing of neighboring elements admits a twist similar to that of twisted geometries in exactly the same sense as the classical prequel to the EPRL model does: 
Consider the piece of a hypersurface intersecting the interior of a spacetime atom; it splits the set of cornercells of the spacetime atom into two complementary subsets. The split is either of type $1-3$ or of type $2-3$. Let us describe the $1-3$ case. In this case the portion of hypersurface has the same combinatorial structure as the hypersurface composed by the gluing of the 4 3-cubes sharing a vertex of a 4-cube. (Its 3d analog is easy to picture just looking at a corner of a room, the 3 2-cubes making the portion of hypersurface are two walls and the flor.) Our $B$ field associates compatible elements to the 2-faces of neighboring 3-cubes, but if the geometricity conditions are imposed softly a shape mismatch is allowed in the same manner as in twisted geometries \cite{FS}.

\subsection{Curvature in spacetime atoms}

As far as the gauge field is concerned, we follow Reisenberger's discretization \cite{Rei1, Rei2} in which spacetime atoms are curved. 
In each corner cell of a 4-simplex we have a measure of curvature in a complete set of independent directions: the gauge field assigns a group element to each elementary plaquette (or wedge in Reisenberger's terminology) $g_{\partial s}$ of the corner cell. 
Then this feature is used to write a regularization of the curvature $F_{ab}^i$ from the continuum to the discrete scenario $\theta_{\partial s}^i$, giving a complete representation of curvature at each corner cell. 
This notion of discrete curvature of the connection was already used in the definition of the action of our model, and its convergence in the continuum limit was proven. It can also be used to regularize other observables (functions in the space of histories) measuring curvature. 

As mentioned above, the contribution to the action of a spacetime atom is more involved that a single term in the action of the classical prequel to the EPRL model, and it should be compared to the contribution to the action of the star of a vertex. 
An advantage of having a handle of curvature inside a spacetime atom over having that information in the star of a vertex is that the connectivity of the internal structure of Reisenberger's atoms is always the same; on the other hand, in a Regge lattice the star of one vertex or another one may have different combinatorial structures. This regularity played an important role in constructing a simple proof of the convergence of the action in the classical continuous/macroscopic limit. 

Another consequence of working with curved spacetime atoms is that, as seen in Section \ref{Model}, the action associated to a region with boundary and corners can be specified without having any information about the complement of the region. 
We mention corners not because we care about regions of the boundary that are not differentialble; we want to consider general gluing of spacetime regions in which more than two spacetime regions can intersect in codimension two cells. 
Notice that the additivity properties of the action for regions with boundary given in Section \ref{Model} is clear even in the case of those arbitrary gluings. 
In contrast, in Regge calculus the deficit angle of a boundary bone has to be arbitrarily split into contributions from the atoms sharing that bone. (An important case where the arbitrariness can be resolved and gluing works \cite{Hartle+Sorkin} is when spacetime regions are only allowed to glue along codimension one hypersurfaces without corners.) 
One could say that Reisenberger's discretization with curved spacetime atoms naturally allows subdivisions into spacetime regions with boundary and corners. This holds in particular for the models presented in this article.

Only very recently have the first spin foam simulations involving more than one simplex been performed \cite{ESFM1, SFmoreThan1simplex1, SFmoreThan1simplex2}. Since the geometry of the current spin foam models is based on flat 4-simplices, the curvature degrees of freedom have barely been tested in spin foam models for quantum gravity \cite{Livine}. In contrast, if a quantum version of our model is constructed, the study of a single spacetime atom would necessarily include curvature degrees of freedom. 

The momenta to the boundary connection $u$ in our model, which describes the boundary geometry, is curvature dependent. Thus, a calculation regarding boundary geometry observables for a single spacetime atom would test curvature degrees of freedom. In other words, a single vertex in the spin foam representation would probe curvature degrees of freedom. 

We end this subsection with a brief discussion comparing Reisenberger's discretization of gauge fields with a discretization in which one could presume that spacetime atoms are flat, like the one naturally emerging from Regge's discretization of spacetime geometry. Is one of these pictures more fundamental? 
First we should add a clarifying note; when we mentioned Regge's discretization we were talking in a figurative manner because the variables used by the EPRL model and the Barrett-Crane model are not a direct quantization of the link's lengths. In particular, the EPRL model uses a gauge field and a $B$ field just as we do in our $B$-field models. The gauge field of the EPRL model can be described using a lattice regularization on the lattice living in the 1-skeleton of the cellular decomposition dual to the triangulation. There is, however, an alternative equivalent picture in which group elements are not associated to lattice links (of the mentioned dual cellular decomposition) but to the codimension 1 hypersurfaces. In this picture parallel transport can be defined for curves ``in the continuum of Regge's triangulation'': The result is trivial when the curve stays inside a 4-simplex, and when it crosses a codimension 1 hypersurface it is multiplied by the left by the corresponding group factor. We will continue this short discussion using the lattice gauge theory picture. 
The issue boils down to comparing two choices regarding the placement of the smallest loops (closed oriented curves), the minimal circuits in which curvature can be sensed. Either they are contained inside spacetime atoms, or spacetime atoms are somehow associated to the vertices visited by the minimal loops. In our view, both pictures provide a cutoff for the field and that is their purpose. 
Another job that they should get done is giving rise to continuum fields in an appropriate continuum/macroscopic limit. In our opinion, Reisenberger's continuum limit \cite{Rei1}, that we adapt to study the continuum limit of our action in Section \ref{ContinuumLimit}, is extremely straight forward in comparison with the continuum limit for Regge's discretization \cite{CheegerEtal}. 

\subsection{Bulk geometry and matter coupling}

The fields of our models lead to a picture of the bulk geometry at every atom. 

Let us explain this claim first for our $e$ model. The geometry of the model is described by a discretized tetrad field $e$. Inside each corner cell $c < \nu$ of each atom we have a soldering form evaluated at the baricenter of the atom $C\nu$. This gives us a picture of the bulk geometry at each atom; in particular, we can write a (pseudo) metric associated to each corner cell. 
For corner cell $c_P < \nu$ the (pseudo) metric is  
$g_{RS}(c_P) = e_R^I (\nu) e_R^J (\nu) \eta_{IJ}$ (with $R,S \neq P$); 
It tells us the inner product of the Minkowski vectors associated to a pair of internal edges $e_{l_R} , e_{l_S}$ in atom $\nu$. 
The five corner cells in an atom have pairwise compatible metrics in the sense that the metrics of two corner cells agree on the three-dimensional spaces where they intersect. 

Having such an $e$ field lets us write every geometric observable of interest, where by an observable we mean a function of the space of histories. 
If in the continuum we have a geometrical observable written as a differential form written in terms of the matter field and the soldering form, we have a corresponding discrete analog. Moreover, if the geometrical meaning of the observable in the continuum comes about in terms of integrals of the differential form, then we can clearly establish convergence of the discretization using the techniques of Section \ref{ContinuumLimit}. In this way we can give regularized expressions for lengths of curves, areas of surfaces, volumes of hypersurfaces and 4-volumes of spacetime regions. Notice that at the level of discrete geometry, the direct regularization that we are talking about deals with curves, surfaces or hypersurfaces intersecting the interior of atoms; in contrast, other models have geometric information about geometric objects that lie at the boundary of spacetime atoms. (Additionally, a picture of the boundary geometry in our model can be obtained using the momentum of the boundary connection $u$.)

Now let us move to our $B$ field models. We can regularize geometric expressions written in the continuum in terms of the $B$ field. We can regularize the area of a surface. We can also regularize the (pseudo) metric density\\
$\tilde{g}_{ab}(x)= \frac{1}{12} f_{ijk} \epsilon^{cdef} B^i_{ac}(x) B^j_{de}(x) B^k_{fb}(x)$ 
\[
\tilde{g}_{RS}(c_{P}<\nu)= \frac{1}{12} f_{ijk} \epsilon^{ABCDP} B^i_{RA}(\nu) B^j_{BC}(\nu) B^k_{DS}(\nu) \quad  \mbox{ with }R\neq P, S\neq P. 
\]

There are two ways to take advantage of this picture of bulk geometry. The first one is to study correlations of geometric observables. The second, which leads to a new perspective, is a natural way to regularize matter couplings. The term in the action describing the coupling would be a sum over atoms (and over corner cells) where the discretized matter field and the bulk geometric picture are combined in an expression that approximates the integral of the continuum expression over the 4-simplex. Below we sketch two examples: a dust field and a Yang-Mills field. 

If we describe a dust field coupled to gravity using a unit length time-like vector field $U$ and a real valued function $\rho$ to describe its spacetime density relative to the volume element, the coupling term in the continuum is $S_{\mbox{dust}}= -\frac{1}{2} \int \rho \tilde{g}_{ab} U^a U^b$ 
\cite{dust}. Its discretization is 
\[
S_{\mbox{dust}, \Delta}= -\frac{1}{2} \sum_{\nu , c<\nu} \rho(\nu) \tilde{g}_{RS}(c) U^R(\nu) U^S(\nu), 
\]
where $\rho(\nu) = \rho (x = C\nu)$, and $U(\nu)$ is defined as follows: 
$U(x)$ together with the volume form $\eta(x)$ induced by the metric determine a 3-form $u_{abc}(x) = U^d(x) \tilde{\eta}(x)$. 
Then this 3-form is integrated in each of the five hypersurfaces $\Sigma_P$ internal to $\nu$ and transversal to the internal edge $l_P$; 
and we define $U^P = \int_{\Sigma_P} u$. 

Instead of describing time-like dust coupled to gravity, we could provide a discretization of null dust coupled to gravity following the same techniques but starting from \cite{nulldust}. Studying the couplings to different types of matter could be a handle to begin the exploration of the microscopic origins of the macroscopic causal structure.

The Yang-Mills field could be coupled to gravity directly at the quantum level. 
Spin foam models for euclidean quenched compact QED and quantum $SU(2)$ Yang-Mills (on a hypercubic discretization) 
that follows from the heat-kernel action \cite{heat-kernel} is sketched in \cite{Rei2}. They are direct extensions of the 2d Yang-Mills model of Gross and Taylor \cite{Gross-Taylor}, and they are naturally written as a product of amplitudes associated to each atom. 
The spin foam model has the same structure as the model for $U(1)$ or $SU(2)$ BF theory (see \cite{Perez1}) but an extra factor of 
$e^{-\frac{e^2}{a^2} n(f)^2 \mbox{Area}(f)}$
or $e^{-\frac{g^2}{a^2} (j(f)(j(f)+1)) \mbox{Area}(f)}$ is associated to each face $f$ of the 2-complex colored with the irreducible representation $n(f)$ or $j(f)$ and with area $\mbox{Area}(f)$. 
In our model we can write expressions for the areas of the wedges, the intersection of the 2-complex with spacetime atoms. 
Thus, it is natural to suggest that with a quantum version of our model we would be able to assign a factor of the weight associated to a history of gravity and Yang-Mills field which expresses their coupling. The resulting path integral would be a spin foam model factorizable as a product over spacetime atoms.

Since one of the distinctive features of our model is that it assigns areas to faces of the 2-complex (and not only to faces of the dual cellular decomposition) and that feature is essential in the mentioned natural coupling to Yang-Mills, we consider that it may be considered as a possible advantage of our model. 
There is a very interesting model coupling the Barrett-Crane spin foam model to Yang-Mills \cite{Oriti-Pfeiffer}. 
It is not a spin foam model because coupling to the geometry provided by the Barrett-Crane model forces the Yang-Mills field to live in the triangulation and not in the dual cellular complex. This would not happen in our proposal.

\section{Summary and further considerations}
\label{Summary}

We presented a discrete model for gravity, the e-field model in Section \ref{Model}, and proved that its action converges in the continuum limit to Holst's action for general relativity.

In our model, there is a copy of Minkowski's space associated with each spacetime atom. When the gravitational field ($e$ or $B$) is not degenerate, one can use it to provide a relative reference frame in that copy of Minkowski's space. (For a expression of the spacetime metric written in terms of the tetrad field $e$ see Subsection \ref{discretization}, and for an expression written in terms of the $B$ field see Subsection \ref{GeometricityEPRL}.) 
The role played by the connection field is to provide a communication system connecting those local relative reference frames.

Striving to provide a new point of view for spin foam models for quantum gravity, we gave two versions of our model, which trade the discrete tetrad field $e$ for a discrete $B$ field (B-field models 1 and 2 in Section \ref{B}). 
The $e$-model and the $B$-field models are defined even when the fields are degenerate, but our proof of classical equivalence holds in the no-degenerate regime only. 
With the first version, the plan is to follow an adapted version of the EPRL quantization \cite{EPRL}. This will be done first at the level of each corner cell of a 4-simplex, and then the pieces will be glued and projected to solve a quantum version of our G3 condition for constructing a variant of the Conrady-Hnybida extension \cite{Con-Hny} of the EPRL vertex. We expect that the intuitive geometric picture for the resulting quantum theory will be a twisted geometry where the 2-surfaces that are glued with a twist are internal to each 4-simplex (the 2-surfaces at the boundary of the intersection between two neighboring corner cells). 
The second version was constructed with the intention of convincing researchers that it can be used to construct an effective spin foam model along the lines proposed by Dittrich et al in \cite{ESFM1}. 
In our view, their use of area-angle Regge calculus holds them back from achieving even more efficient simulations. A classical model formulated over a Reisenberger discretization would be readily compatible with boundary spin network states and would have superior gluing properties. One option is to use the model presented in the article. 
A third quantization avenue would be a direct path integral quantization of our tetradic model by considering a 4-simplex $\nu$ with a fixed connection on its boundary (the $k_r$ for $r < \partial \nu$) and integrating out the bulk degrees of freedom, which are the bulk connection and the tetrad (the $h_l$ and the $e_l$ for $l < \nu$) using an appropriate measure. 
It is natural to ask how do these three quantization avenues compare to each other. This type of study was proposed in \cite{Engle-Thiemann-Han}, but since our model has an actual tetrad field at the discrete level, it may provide an interesting new perspective. 

In the previous section we mentioned two features of our model that could open new perspectives: 
First, the fact that our model uses spacetime atoms with curvature means that a quantum version of our model would test curvature degrees of freedom even at the level of a single spin foam vertex. 
Second, a distinctive ``locality property'' of our model is that it provides a picture of bulk geometry at the level of spacetime atoms that yields expressions for matter coupling in which the coupling term of the action splits as a sum over spacetime atoms.

Now we briefly address a central issue in the context of our work: why restricting our study to Holst's action. 

The action that we discretized \eqref{holst} is Holst's action for General Relativity, but in the same field space there are other actions which are also compatible with Einstein's equations. For instance, one can add topological terms to \eqref{holst}. For a discussion see \cite{AddingTopologicalTerms}. In the continuum, those terms do not affect the field equations. If we discretize those terms, however, even if we prove convergence of the action in the continuum limit, it is likely that they would affect the discrete field equations. A framework in which topological invariants of gauge fields can be accounted for in a lattice approach without taking the continuum limit is given in \cite{Zapata-Meneses}. 
Including a tetrad field or a $B$ field in that framework extending standard lattice fields would need some work, but the central idea is the same: include the relevant homotopy data. For a first study of sigma models see \cite{Dall'Olio-Zapata}.

Considering a more general action is also natural in the broader context of constructing a quantum field theory using Wilsonian renormalization to get rid of the ambiguities associated with a cutoff. See \cite{AddingTopologicalTerms} for a discussion. 

In spacetime regions with boundary and corners boundary terms do affect the geometro-algebraic structure of the classical theory and its quantization. This is another motivation for a future consideration of a more general action. For a discussion see \cite{ReviewCornerModesETC}. 

\section*{Acknowledgements}
	We thank M. Reisenberger for fruitful discussions. 

	\medskip

\bibliographystyle{unsrt}

\end{document}